\documentclass{ws-ijgmmp}

\begin{document}

\markboth{J.A. Belinch\'on, C. Gonz\'alez, S. Dib}
{Instructions for Typing Manuscripts (Paper's Title)}

%
\catchline{}{}{}{}{}
%

\title{Self-similar cosmological solutions in $f(R,T)$ gravity theory}

\author{Jos\'{e} Antonio Belinch\'{o}n}

\address{Dpt. Matem\'{a}ticas, Facultad de Ingenier\'{\i}a, Universidad de Atacama, Av.
Copayapu 485, Copiap\'{o}, Chile\\
\email{jose.belinchon@uda.cl }}

\author{Carlos Gonz\'alez}

\address{Dpt. F\'{\i}sica, Facultad Ciencias Naturales, Universidad de Atacama, Av.
Copayapu 485, Copiap\'{o}, Chile\\
carlos.gonzalez@uda.cl }

\author{Sami Dib}

\address{Max Planck Institute for Astronomy, K\"{o}nigstuhl 17, 69117,
Heidelberg, Germany\\
sami.dib@gmail.com}

\maketitle

\begin{history}
\received{(Day Month Year)}
\revised{(Day Month Year)}
\end{history}

\begin{abstract}
We study the $f(R,T)$ cosmological models under the self-similarity hypothesis.
We determine the exact form that each physical and geometrical quantity may
take in order that the Field Equations (FE) admit exact self-similar solutions
through the matter collineation approach. We study two models: the
case$\ f(R,T)=f_{1}(R)+f_{2}(T)$ and the case $f(R,T)=f_{1}(R)+f_{2}%
(R)f_{3}(T)$. In each case, we state general theorems which determine
completely the form of the unknown functions $f_{i}$ such that the field
equations admit self-similar solutions. We also state some corollaries as
limiting cases. These results are quite general and valid for any homogeneous self-similar
metric$.$ In this way, we are able to generate new cosmological scenarios. As
examples, we study two cases by finding exact solutions to these particular models.

\end{abstract}


\section{Introduction}

The physical and mathematical importance of the modified gravitational models
has recently been pointed out by several authors (see, for instance, \cite{N-O
10} \cite{Tim} \cite{Amendola}, \cite{Faraoni STT}, and \cite{NO}),
respectively, since this kind of theories explain in a better way the dynamics
of the very early universe, as well as its current acceleration. Obviously,
although such class of theories are more general than the usual
Jordan-Brans-Dicke models, they may be generalized in order to incorporate
corrections to the Ricci scalar term as the $f(R)$ models (see for example
\cite{Faraoni}).

Recently, a new theory, known as $f(R,T)$ gravity, has been formulated in
\cite{H1}. In the context of $f(R,T)$ modified theories of gravity, the
gravitational Lagrangian is given by an arbitrary function of the Ricci scalar
$R$ and of the trace of the energy-momentum tensor $T$. Hence, $f(R,T)$
gravity theories can be interpreted as alternatives of the $f(R,L_{m})$ type
gravity theories \cite{Lm}, with the gravitational action depending not on the
matter Lagrangian, but on the trace of the matter energy-momentum tensor. The
gravitational field equations in the metric formalism show that the field
equations explicitly depend on the nature of the matter source. The dependence
on $T$ may be induced by exotic imperfect fluids, or quantum effects
(conformal anomaly) \cite{H1}.

This theory has received recently a great amount of attention. For extensive
reviews of $f(R,T)$ gravity see \cite{H2,H3,H4,H5}. In the following, we
discuss some of the recent relevant works in the field. Houndjo \cite{AP1}
studied the cosmological reconstruction of $f(R,T)$ gravity describing
matter-dominated and accelerated phases. Special attention was paid to the
specific case of $f(R,T)=f_{1}(R)+f_{2}(T)$. In \cite{AP2}, the authors have
shown that the dust fluid reproduces $\Lambda$CDM cosmology. In \cite{AP3},
the authors discussed a non-equilibrium picture of thermodynamics at the
apparent horizon of an Friedmann-Lemaitre-Robertosn-Walker (FLRW) universe in
$f(R,T)$ gravity, and the validity of the first and second law of
thermodynamics in this scenario were checked. It was shown that the Friedmann
equations can be expressed in the form of the first law of thermodynamics and
that the second law of thermodynamics holds both in the phantom and
non-phantom phases. The energy conditions have also been extensively explored
in $f(R,T)$ gravity, for instance, in an FLRW universe with perfect fluid
\cite{AP4}, and in the context of exact power-law solutions \cite{AP5}. It was
found that certain constraints have to be met in order to ensure that power
law solutions may be stable, and match the bounds prescribed by the energy
conditions. In \cite{AP6}, it was also shown that the energy conditions were
satisfied for specific models. Furthermore, an analysis of the perturbations
and stability of de Sitter solutions and power-law solutions was performed,
and it was shown that for those models in which the energy conditions are
satisfied, de Sitter solutions and power-law solutions may be stable.

The gravitational baryogenesis epoch in the context of $f(R,T)$ theory was
studied in \cite{AP7}. In particular, the following two models,
$f(R,T)=R+\alpha T+\beta T^{2}$ and $f(R,T)=R+\mu R^{2}+\lambda T$, were
considered, with the assumption that the universe is filled by dark energy and
perfect fluid. Moreover, it was assumed that the baryon to entropy ratio
during the radiation domination era was non-zero. In \cite{AP8} a particular
form for $f(R,T)=R+\lambda T$ was adopted, and inflationary scenarios were
investigated. By taking the slow-roll approximation, interesting models were
obtained, with a Klein-Gordon potential leading to observationally consistent
inflation observables. These results make it clear-cut that, in addition to
the Ricci scalar and scalar fields, the trace of the energy momentum tensor
also plays a major role in driving inflationary evolution. The slow-roll
approximation of cosmic inflation within the context of $f(R,T)$ gravity was
studied in \cite{AP9}. A gravitational model unifying both $f\left(
R,L_{m}\right)  $ and $f(R,T)$ theories into the $f\left(  R,L_{m},T\right)  $
was developed in \cite{HaHa}.

In all the studied models, it is necessary to impose a specific form for the
function $f(R,T),$ where usually $f(R,T)=R+2f(T)=R+\lambda T$, and the matter
content is a perfect fluid. Therefore, it would be highly desirable to have a
general method for determining the functional form of $f(R,T)$ in such a way
that the resulting FE are integrable (in the Louville form). Such an approach
would also allow us to make definite cosmological predictions, and to compare
theory with observations. Moreover, it would be important to determine the
form the different physical and geometrical quantities may take in order that
the FE could be integrated. Therefore, it would be necessary to have a
fundamental method according to which the form (or forms) of the geometrical
and physical quantities could be fixed, and, if possible, to calculate exact
solutions for the proposed models. In this context, several geometric methods
can be used, such as, the matter collineation (self-similar solutions) approach.

For the matter content in the $f(R,T)$ gravity it is generally assumed that it
consists of a fluid that can be characterized by two thermodynamic parameters
only, the energy density and the pressure, respectively. To solve the
resulting FE it is necessary to consider particular classes of $f(R,T)$%
-modified gravity models, obtained by explicitly specifying the functional
form of $f$. Generally, the field equations also depend on the physical nature
of the matter field. Hence, in the case of $f(R,T)$ gravity, depending on the
nature of the matter source, for each choice of $f$, we can obtain several
theoretical models, corresponding to different matter models.

The study of self-similar (SS) models is quite important since, as it has been
pointed out in \cite{RJ}, they correspond to equilibrium points. Therefore, a
large class of orthogonal spatially homogeneous models are asymptotically
self-similar at the initial singularity, and are approximated by exact perfect
fluid, or vacuum self-similar power law models.

Exact self-similar power-law models can also approximate general Bianchi
models at intermediate stages of their evolution \cite{Coley}. From the
geometrical point of view, self-similarity is defined by the existence of a
homothetic vector field $\mathcal{H}$ in the spacetime, which satisfies the
equation $L_{\mathcal{H}}g_{\mu\nu}=2\alpha g_{\mu\nu}$ \cite{CC}-\cite{CC1}.
The geometry and physics at different points on an integral curve of a
homothetic vector field (HVF) differ only by a change in the overall length
scale, and in particular, any dimensionless scalar will be constant along the
integral curves. Self-similarity also plays an important role in the study of stellar structure \cite{Mak}.

Therefore, the aim of this paper is to study $f(R,T)$ cosmological models by
using geometrical methods in order to determine the form of the physical
quantities, as well as the other unknown functions that appear in the field
equations. In particular, we are interested in studying whether self-similar
solutions do exist, and how must each physical quantity behaves in order that
the FE admit such class of solutions. We also show how to use this approach in
order to generate additional solutions.

The present paper is organized as follows. In Section II we introduce first
the basic gravitational model, and we outline the field equations. Then we
determine the exact form that each physical quantity may take in order for the
FE to admit exact self-similar solutions through the matter collineation
approach. Two models are explicitly studied. In Section III, we study the
case$\ f(R,T)=f_{1}(R)+f_{2}(T)$, and in Section IV, we explore the case $f(R,T)=f_{1}%
(R)+f_{2}(R)f_{3}(T)$ is studied. In each case, we state a general theorem
which determines completely the form of the unknown functions $f_{i}$ for the
field equations to admit self-similar solutions. We also state some
corollaries as limiting cases. The results are quite general and valid for any
self-similar metric. Section V is devoted to the investigation of exact
solutions to two particular models. In Section VI, we discuss our findings and
conclude our work. In the Appendix, we show how to apply this geometrical
method to obtain other types of solutions, such as the exponential one.

\section{Field equations of $f(R,T)$ gravity, and self-similarity}

In this Section we briefly introduce the field equations of the
$f(R,T)$ gravity theory, and we outline the basic concepts of self-similarity.

\subsection{Field equations}

For a general function $f=f(R,T)$, the FE and the matter conservation equation
read (see for instance \cite{H1}, and \cite{ConOK}-\cite{ConOK2} for the
correct conservation equation),
\begin{align}
R_{\mu\nu}-\frac{1}{2}Rg_{\mu\nu}  &  =\frac{1}{f_{R}}\left(  \left(
8\pi-f_{T}\right)  T_{\mu\nu}-f_{T}\Theta_{\mu\nu}+\frac{1}{2}\left(
f-f_{R}R\right)  g_{\mu\nu}-\left(  g_{\mu\nu}\square-\nabla_{\mu}\nabla_{\nu
}\right)  f_{R}\right)  ,\label{FE1}\\
\nabla^{\mu}T_{\mu\nu}  &  =\frac{f_{T}}{8\pi-f_{T}}\left[  \left(  T_{\mu\nu
}+\Theta_{\mu\nu}\right)  \nabla^{\mu}\ln f_{T}+\nabla^{\mu}\Theta_{\mu\nu
}-\frac{1}{2}g_{\mu\nu}\nabla^{\mu}T\right]  , \label{FE2}%
\end{align}
where%
\begin{equation}
\Theta_{\mu\nu}=-2T_{\mu\nu}+g_{\mu\nu}L_{\mathrm{m}}-2g^{\alpha\beta}%
\frac{\partial^{2}L_{\mathrm{m}}}{\partial g^{\mu\nu}\partial g^{\alpha\beta}%
},
\end{equation}
and $L_{\mathrm{m}}$ stands for the Lagrangian of the matter field. We
consider as matter model a perfect fluid, and we take $L_{\mathrm{m}}=-p$
\cite{Lagran0}-\cite{Lagran1}, thus obtaining
\begin{equation}
\Theta_{\mu\nu}=-2T_{\mu\nu}-pg_{\mu\nu}\,,\qquad T_{\mu\nu}=\left(
\rho+p\right)  u_{\mu}u_{\nu}+pg_{\mu\nu}\,,\qquad T=-\rho+3p,\,\qquad
p=\omega\rho.
\end{equation}

\subsection{Self-similarity and $f(R,T)$ gravity}

Self-similarity is defined by the existence of a homothetic vector field ${V}$
in the spacetime (\cite{CT, 22}, which satisfies the condition
\begin{equation}
L_{V}g_{\mu\nu}=2\alpha g_{\mu\nu}, \label{gss1}%
\end{equation}
where $g_{\mu\nu}$ is the metric tensor, $L_{V}$ denotes Lie differentiation
along the vector field ${V}\in\mathfrak{X}(M)$ and $\alpha$ is a constant (see
for general reviews \cite{CC} and \cite{Hall}, respectively).

If we consider the Einstein equations in standard general relativity
$G_{\mu\nu}=8\pi GT_{\mu\nu},$ where $T_{\mu\nu}$ is an effective
stress-energy tensor, then if the spacetime is homothetic, the energy-momentum
tensor of the matter fields must satisfy $L_{V}T_{\mu\nu}=0$. Nevertheless, in
this work, we are not interested in finding the set of vector fields
$V\in\mathfrak{X}(M),$ that satisfy the standard general relativistic
condition. If the homothetic vector field in known (HVF) $\mathcal{H}$, that
is, the vector field satisfying the condition $L_{\mathcal{H}}g_{\mu\nu
}=2g_{\mu\nu}$, (see, for example \cite{CC}), then $\mathcal{H}$ is also a
matter collineation, i.e., it satisfies the condition $L_{\mathcal{H}}%
T_{\mu\nu}=0$. In the following, we extend this condition to the $f(R,T)$
theory to determine the behavior of the main geometrical and physical
quantities so that the field equations admit self-similar solutions (see
\cite{Hall}).

Therefore, in a modified gravity theory we impose generally the condition
\begin{equation}
L_{\mathcal{H}}T_{\mu\nu}^{\mathrm{eff}}=0,
\end{equation}
where $\mathcal{H}$ is a homothetic vector field (HVF), i.e. it verifies the
equation: $L_{\mathcal{H}}g_{\mu\nu}=2g_{\mu\nu},$ for a given metric tensor,
and where $T_{\mu\nu}^{\mathrm{eff}}$ is the effective energy-momentum tensor
of the theory. In \cite{Tony2} it was shown that it is enough to calculate
$L_{\mathcal{H}}\,^{(i)}T_{\mu\nu}^{\mathrm{eff}}=0$, for each component of
the energy-momentum tensor.

For simplicity, in the following we consider the flat, homogeneous and
isotropic FLRW metric only,
\begin{equation}
ds^{2}=dt^{2}-a^{2}(t)\left(  dx^{2}+dy^{2}+dz^{2}\right)  ,
\end{equation}
where $a$ is the scale factor. Thus the HVF yields (see for instance
\cite{HW}),%
\begin{equation}
\mathcal{H}=t\partial_{t}+\left(  1-a_{1}\right)  \left(  x\partial
_{x}+y\partial_{y}+z\partial_{z}\right)  ,
\end{equation}
where $a_{1}\in\mathbb{R},$ is a numerical constant.

Note at this point that the scale factor must behave as $a(t)=t^{a_{1}},$
$a_{1}\in\mathbb{R}^{+}$. We adopt such a simplification since, as we have
shown in \cite{Tony2}, all the physical quantities are homogeneous, that is,
they only depend on the time $t$. Hence, the unique equation of
$L_{\mathcal{H}}^{(i)}T_{\mu\nu}=0,$ that is interesting for us, is the one
corresponding to the temporal coordinate $t\partial_{t}.$ Nevertheless, the
results are also valid for any homogeneous self-similar metric (Bianchi
models, for example) since their homothetic vector field always takes the from
$\mathcal{H}=t\partial_{t}+...$. (see for instance \cite{WE}). This kind of
models are called spatially homogeneous, that is, the $G_{6}$ FLRW. The $G_{4}$
models. If the $G_{4}$ has a subgroup $G_{3}$ which acts simply transitively
on the three-dimensional orbits we obtain the LRS Bianchi models. Otherwise,
we obtain the Kantowski-Sachs models, and the SS Bianchi models. In addition, although the model evolves in time, the evolution is fully specified by this symmetry property, and the FE are purely algebraic when expressed in terms of expansion-normalized variables.

Then the FE of the $f(R,T)$ theory do admit self-similar solutions if
\begin{equation}
L_{\mathcal{H}}\left(  T^{eff}\right)  =0,\qquad T^{eff}=\frac{1}{f_{R}%
}\left(  \left(  8\pi-f_{T}\right)  T_{\mu\nu}-f_{T}\Theta_{\mu\nu}+\frac
{1}{2}\left(  f-f_{R}R\right)  g_{\mu\nu}-\left(  g_{\mu\nu}\square
-\nabla_{\mu}\nabla_{\nu}\right)  f_{R}\right)  ,
\end{equation}
where $L_{\mathcal{H}}$ is the Lie derivative along the HVF (see for instance
\cite{Tony2} and \ \cite{Tony3} where this method was applied to other
gravitational theories).

\section{Case 1: $f(R,T)=f_{1}(R)+f_{2}(T)$}

Under the hypothesis $f(R,T)=f_{1}(R)+f_{2}(T),$ the FE (\ref{FE1}-\ref{FE2})
take the form
\begin{equation}
G_{\mu\nu}=\frac{1}{f_{1R}}\left(  \left(  8\pi+f_{2T}\right)  T_{\mu\nu
}+f_{2T}pg_{\mu\nu}+\frac{1}{2}f_{2}g_{\mu\nu}+\frac{1}{2}\left(  f_{1}%
-f_{1R}R\right)  g_{\mu\nu}-\left(  g_{\mu\nu}\square-\nabla_{\mu}\nabla_{\nu
}\right)  f_{1R}\right)  . \label{FEC1}%
\end{equation}

We state the following Theorem.

\begin{theorem}
\label{Teorema1} The FE (\ref{FEC1}) admit SS solutions if and only if
$f_{1}=R^{n},$ and $f_{2}=\lambda T,$ with $n,\lambda\in\mathbb{R}.$
\end{theorem}

\begin{proof}
We start by proving that $f_{1}=R^{n}.$ Thus, from $L_{\mathcal{H}}\left(
T_{\mu\nu}^{R}\right)  =0,$ that is,%
\begin{equation}
T_{\mu\nu}^{R}=\frac{1}{2}\left(  F^{-1}f_{1}-R\right)  g_{\mu\nu}%
+F^{-1}\left(  \nabla_{\mu}\nabla_{\nu}-g_{\mu\nu}\square\right)  F,\qquad
F=f_{1R},
\end{equation}
then%
\begin{equation}
L_{\mathcal{H}}\left(  \frac{1}{2}\left(  F^{-1}f_{1}-R\right)  g_{\mu\nu
}\right)  =0,
\end{equation}
we calculate only the first component, $L_{\mathcal{H}}\left(  T_{11}\right)
=0,$ since with the other components we only obtain restrictions on the scale
factor. Thus, this yields%
\begin{equation}
t\left(  f_{1}^{\prime}-R^{\prime}F-F^{\prime}f_{1}F^{-1}\right)
=-2(f_{1}-FR),
\end{equation}
where $^{\prime }=d/dt,$ which may be re-written in the following form%
\begin{equation}
t\left(  -f_{1R}^{-1}f_{1RR}R^{\prime}f_{1}\right)  +2(f_{1}-f_{1R}R)=0,
\end{equation}
since
\begin{equation}
f_{1}^{\prime}=\frac{df_{1}}{dR}R^{\prime}=FR^{\prime},\qquad F^{\prime
}=f_{1RR}R^{\prime}.
\end{equation}

Now, by taking into account that $tR^{\prime}\thickapprox-2R,$ then, the
equation $L_{\mathcal{H}}\left(  T_{11}\right)  =0$ yields%
\begin{equation}
f_{1RR}=\frac{f_{1R}^{2}}{f}-\frac{f_{1}}{R}\qquad\Longrightarrow\qquad
f_{1}=C_{1}R^{n},\qquad n\in\mathbb{R},
\end{equation}
where $C_{1}$ a constant of integration.

Now, we calculate%
\begin{equation}
L_{\mathcal{H}}\left(  F^{-1}\left(  \nabla_{\mu}\nabla_{\nu}-g_{\mu\nu
}\square\right)  F\right)  =0,
\end{equation}
where the first component of the above equation, $L_{\mathcal{H}}\left(
T_{11}\right)  =0,$ yields
\begin{equation}
t\left(  F^{\prime\prime}H+F^{\prime}H^{\prime}-\frac{F^{\prime2}}{F}H\right)
=-2HF^{\prime},
\end{equation}
where $H=a^{\prime}/a,$ but we already know that $H=h_{0}t^{-1}$ (within the
self-similar approach), so we rewrite this as%
\begin{equation}
F^{\prime\prime}=\frac{F^{\prime2}}{F}-\frac{F^{\prime}}{t}\qquad
\Longrightarrow\qquad F=F_{0}t^{C_{1}},\qquad C_{1}\in\mathbb{R}.
\end{equation}
Therefore, if we set $C_{1}=2\left(  1-n\right)  $ then we arrive at the
conclusion that $F=F_{0}t^{2(1-n)}=nR^{n-1}.$ Note that $R=R_{0}t^{-2}.$

Now we calculate
\begin{equation}
L_{\mathcal{H}}\left(  T_{\mu\nu}^{T}\right)  =L_{\mathcal{H}}\left(
F^{-1}\left(  \left(  8\pi+f_{2T}\right)  T_{\mu\nu}+f_{2T}pg_{\mu\nu}%
+\frac{1}{2}f_{2}g_{\mu\nu}\right)  \right)  =0,
\end{equation}
where it will be sufficient to prove that
\begin{equation}
L_{\mathcal{H}}\left(  F^{-1}T_{\mu\nu}\right)  =0,\qquad L_{\mathcal{H}%
}\left(  F^{-1}f_{2}g_{\mu\nu}\right)  =0,
\end{equation}
so, from $L_{\mathcal{H}}\left(  F^{-1}T_{11}\right)  =0,$ we get that%
\begin{equation}
\frac{\rho^{\prime}}{\rho}-\frac{F^{\prime}}{F}=-\frac{2}{t}\qquad
\Longrightarrow\qquad\rho F^{-1}=\rho_{0}t^{-2},
\end{equation}
where $\rho_{0}$ is a constant of integration.

In the same way, we obtain a similar behavior for $p,$ that is, $pF^{-1}%
=p_{0}t^{-2}.$ By taking into account our previous result, $F=F_{0}%
t^{2(1-n)},$ then%
\begin{equation}
\rho=\rho_{0}t^{-2n},
\end{equation}
which implies that
\begin{equation}
T=-\rho+3p=T_{0}t^{-2n}.
\end{equation}

Now from $L_{\mathcal{H}}\left(  F^{-1}f_{2}g_{\mu\nu}\right)  =0,$%
\begin{equation}
\frac{f_{2}^{\prime}}{f}-\frac{F^{\prime}}{F}=-\frac{2}{t},\qquad
f_{2}^{\prime}=f_{2T}T^{\prime},\qquad\Longrightarrow\qquad f_{2}F^{-1}%
=t^{-2},
\end{equation}
thus we may set%
\begin{equation}
f_{2}=\lambda T,\qquad\lambda\in\mathbb{R},
\end{equation}
as it is required.

Note that with these results, the rest of equations, that is, $L_{\mathcal{H}%
}\left(  F^{-1}f_{2T}T_{\mu\nu}\right)  =0,$ and $L_{\mathcal{H}}\left(
F^{-1}f_{2T}pg_{\mu\nu}\right)  =0,$ are trivially satisfied.
\end{proof}

\bigskip

Therefore,\textbf{ Theorem 1} states that $f_{1}(R)=R^{n},$ $n\geq1,$ in the
framework of the self-similar solution, the scale factor (or scale factors for
the self-similar Bianchi models, for example) follows a power-law $a=t^{a_{1}%
}$, which implies that the Ricci curvature (scalar curvature) behaves as
$R=R_{0}t^{-2}.$ We summarize the result below in the form
\begin{align}
f(R,T)  &  =R^{n}+\lambda T,\qquad R=R_{0}t^{-2},\qquad a=t^{a_{1}},\qquad
a_{1}\in\mathbb{R}^{+},\nonumber\\
p  &  =\omega\rho=p_{0}t^{-2n},\qquad T=-\rho+3p=T_{0}t^{-2n},\qquad
n,\lambda\in\mathbb{R}. \label{R1}%
\end{align}

This result is new, and generalizes previous works on self-similar
cosmological models.

From the above \textbf{Theorem 1}, we may deduce the following trivial
corollary, in which we reobtain the first $f(R,T)$ type model proposed by
Harko et al in \cite{H1}.

\begin{corollary}
\label{Colo1}If $n=1,$ then $f_{2}=\lambda T,$ therefore
\begin{equation}
f(R,T)=R+\lambda T.
\end{equation}

\end{corollary}

\section{Case 2: $f(R,T)=f_{1}(R)+f_{2}(R)f_{3}(T)$}

Under the assumption $f(R,T)=f_{1}(R)+f_{2}(R)f_{3}(T),$ the FE (\ref{FE1}%
-\ref{FE2}) read%
\begin{equation}
G_{\mu\nu}=\frac{1}{K}\left(  \left(  8\pi+f_{T}\right)  T_{\mu\nu}%
+f_{T}pg_{\mu\nu}+\frac{1}{2}f_{2}f_{3}g_{\mu\nu}+\frac{1}{2}\left(
f_{1}-KR\right)  g_{\mu\nu}-\left(  g_{\mu\nu}\square-\nabla_{\mu}\nabla_{\nu
}\right)  K\right)  , \label{FEC2}%
\end{equation}
where $K=f_{1R}+f_{2R}f_{3}$ and $f_{T}=f_{2}f_{3T}.$ We state the following Theorem.

\begin{theorem}
\label{Teorema2} The FE (\ref{FEC2}) admit self-similar solutions if and only
if $f(R,T)=f_{1}(R)+f_{2}(R)f_{3}(T)=R^{n}+R^{n(1-l)}T^{l},$ with
$n,l\in\mathbb{R}.$
\end{theorem}

\begin{proof}
We define%
\begin{align}
T_{\mu\nu}^{eff}  &  =\frac{1}{K}\left(  \left(  8\pi+f_{T}\right)  T_{\mu\nu
}+f_{T}pg_{\mu\nu}+\frac{1}{2}f_{2}f_{3}g_{\mu\nu}+\frac{1}{2}\left(
f_{1}-KR\right)  g_{\mu\nu}-\left(  g_{\mu\nu}\square-\nabla_{\mu}\nabla_{\nu
}\right)  K\right)  ,\\
K  &  =f_{1R}+f_{2R}f_{3},\qquad f_{T}=f_{2}f_{3T},
\end{align}
and also consider
\begin{align}
T_{\mu\nu}^{1}  &  =\frac{1}{K}T_{\mu\nu},\qquad T_{\mu\nu}^{2}=\frac{f_{T}%
}{K}T_{\mu\nu},\qquad T_{\mu\nu}^{3}=\frac{f_{T}}{K}pg_{\mu\nu},\qquad
T_{\mu\nu}^{4}=\frac{f_{2}f_{3}}{K}g_{\mu\nu},\\
T_{\mu\nu}^{5}  &  =\frac{f_{1}}{K}g_{\mu\nu},\qquad T_{\mu\nu}^{6}=Rg_{\mu
\nu},\qquad T_{\mu\nu}^{7}=\frac{1}{K}\left(  g_{\mu\nu}\square-\nabla_{\mu
}\nabla_{\nu}\right)  K
\end{align}
where $R$, by hypothesis, behaves as $R\thickapprox t^{-2}.$

Then, from%
\begin{align}
T_{\mu\nu}^{1}  &  =\frac{1}{K}T_{\mu\nu},\qquad L_{\mathcal{H}}\left(
\frac{1}{K}T_{11}\right)  =\frac{\rho^{\prime}}{\rho}-\frac{K^{\prime}}%
{K}=-\frac{2}{t},\qquad\frac{\rho}{K}=t^{-2},\\
T_{\mu\nu}^{2}  &  =\frac{f_{T}}{K}T_{\mu\nu},\qquad L_{\mathcal{H}}\left(
\frac{f_{T}}{K}T_{11}\right)  =\frac{\rho^{\prime}}{\rho}+\frac{f_{T}^{\prime
}}{f_{T}}-\frac{K^{\prime}}{K}=-\frac{2}{t},\qquad\frac{f_{T}}{K}\rho
=t^{-2},\\
T_{\mu\nu}^{3}  &  =\frac{f_{T}}{K}pg_{\mu\nu},\qquad L_{\mathcal{H}}\left(
\frac{f_{T}}{K}pg_{11}\right)  =\frac{p^{\prime}}{p}+\frac{f_{T}^{\prime}%
}{f_{T}}-\frac{K^{\prime}}{K}=-\frac{2}{t},\qquad\frac{f_{T}}{K}p\thickapprox
t^{-2},\\
T_{\mu\nu}^{4}  &  =\frac{f_{2}f_{3}}{K}g_{\mu\nu},\qquad L_{\mathcal{H}%
}\left(  \frac{f_{2}f_{3}}{K}g_{11}\right)  =\frac{f_{2}^{\prime}}{f_{2}%
}+\frac{f_{3}^{\prime}}{f_{3}}-\frac{K^{\prime}}{K}=-\frac{2}{t},\qquad
\frac{f_{2}f_{3}}{K}\thickapprox t^{-2},\\
T_{\mu\nu}^{5}  &  =\frac{f_{1}}{K}g_{\mu\nu},\qquad L_{\mathcal{H}}\left(
\frac{f_{1}}{K}g_{11}\right)  =\frac{f_{1}^{\prime}}{f_{1}}-\frac{K^{\prime}%
}{K}=-\frac{2}{t},\qquad\frac{f_{1}}{K}\thickapprox t^{-2},
\end{align}
and the trivial (by hypothesis) result
\begin{equation}
T_{\mu\nu}^{6}=Rg_{\mu\nu},\qquad L_{\mathcal{H}}\left(  Rg_{11}\right)
=\frac{R^{\prime}}{R}=-\frac{2}{t}\qquad R\thickapprox t^{-2},
\end{equation}
so this means that
\begin{equation}
\frac{\rho}{K}\thickapprox t^{-2},\qquad\frac{f_{T}}{K}\rho\thickapprox
t^{-2},\qquad\frac{f_{T}}{K}p\thickapprox t^{-2},\qquad\frac{f_{2}f_{3}}%
{K}\thickapprox t^{-2},\qquad\frac{f_{1}}{K}\thickapprox t^{-2},
\end{equation}
then%
\begin{equation}
f_{T}=f_{2}f_{3T}=const.,\qquad f_{1}\thickapprox f_{2}f_{3}\thickapprox
\rho\thickapprox p,\qquad f_{1R}\thickapprox f_{2R}f_{3}.
\end{equation}

Note that $\thickapprox$ means that the quantities have the same order of
magnitude. Therefore, we are only able to obtain relationships between the
physical quantities, but we do not know anything about their individual behavior.

We need to assume the following hypothesis (within the self-similar framework)%
\begin{equation}
\rho\thickapprox p\thickapprox T\thickapprox t^{-x},\text{\qquad}%
x\in\mathbb{R}^{+}%
\end{equation}
so that
\begin{equation}
f_{1}(R)=R^{n},\qquad f_{2}(R)=R^{m},\qquad f_{3}(T)=T^{l},\qquad
n,m,l\in\mathbb{R},
\end{equation}
where $R\thickapprox t^{-2}.$ Therefore%
\begin{equation}%
\begin{array}
[c]{ll}%
f_{2}f_{3T}=const.: & -2m-x(l-1)=0\\
f_{1R}\thickapprox f_{2R}f_{3}: & -2\left(  n-1\right)  =-2(m-1)-xl\\
f_{1}\thickapprox f_{2}f_{3}: & -2n=-2m-xl
\end{array}%
\begin{array}
[c]{c}%
x=\frac{2m}{1-l}\\
2n=2m+xl\\
2n=2m+xl
\end{array}
\end{equation}

Now, from $f_{1}\thickapprox\rho\thickapprox p,$ we obtain
\begin{equation}
f_{1}\thickapprox t^{-2n}\thickapprox\rho\thickapprox t^{-x},\qquad
\Longrightarrow\qquad x=2n
\end{equation}
and therefore%
\begin{equation}
x=\frac{2m}{1-l}=2n,\qquad\Longrightarrow\qquad2n=2m+xl\qquad\Longrightarrow
\qquad m=n(1-l)\qquad x=2n,
\end{equation}
or%
\begin{equation}
l=-\frac{1}{n}\left(  m-n\right)  .
\end{equation}

Hence the model reduces to%
\begin{align}
f(R,T)  &  =f_{1}(R)+f_{2}(R)f_{3}(T)=R^{n}+R^{n(1-l)}T^{l},\qquad
n,l\in\mathbb{R},\label{R2a}\\
f(R,T)  &  =f_{1}(R)+f_{2}(R)f_{3}(T)=R^{n}+R^{m}T^{1-\frac{m}{n}},\qquad
n,m\in\mathbb{R}, \label{R2b}%
\end{align}
where%
\begin{equation}
\rho\thickapprox p\thickapprox T\thickapprox t^{-2n},\qquad R\thickapprox
t^{-2},
\end{equation}
as it is required.
\end{proof}

\bigskip

In conclusion, the above Theorem states that if $f(R,T)=R^{n}+R^{n(1-l)}T^{l}%
$, then the FE admit self-similar solutions. For cosmological reasons, we
usually consider $n\geq1$ and $l\leq1$. Note that if we fix $l=1$ then the
model reduces to the $f(R,T)=f_{1}(R)+f_{2}(T)$ case. From the above Theorem,
we have the following corollaries.

\begin{corollary}
\label{Coro2}If $n=1,$ the model reduces to%
\begin{equation}
f(R,T)=f_{1}(R)+f_{2}(R)f_{3}(T)=R+R^{m}T^{l},\qquad/\qquad m+l=1.
\end{equation}

\end{corollary}

Note that from the proof of the above Theorem, it is easy to see that
\begin{equation}
f_{2}f_{3T}=const.,\qquad f_{1}=R\thickapprox f_{2}f_{3}\thickapprox
\rho\thickapprox p\thickapprox T\thickapprox t^{-2}.
\end{equation}

\begin{corollary}
\label{Coro3}If $n=0,$ the model collapses to%
\begin{equation}
f(R,T)=f_{2}(R)f_{3}(T)=R^{m}T^{l},\qquad x\left(  1-l\right)  =2m.
\end{equation}

\end{corollary}

As we can easily observe, for this last case,%
\begin{equation}
f_{2}(R)=R^{m}=t^{-2m},\qquad f_{2}(T)=T^{l}=t^{-xl},\qquad m,l,x\in
\mathbb{R},
\end{equation}
where $\rho\thickapprox p\thickapprox T\thickapprox t^{-x},$ with the
constraint: $x\left(  1-l\right)  =2m.$ Thus, if we set $x=2,$ then $m+l=1,$
since $f_{2}f_{3T}=const.$

\section{Some exact solutions}

The above theoretical results only determine the behavior of the physical and
geometrical quantities. They state the order of magnitude of the different
quantities, but they do not say anything about the constants $a_{1},$
$\rho_{0}$, etc. To find an exact analytical solution that is also
cosmologically relevant, it is necessary to determine them in such a way that
we may make predictions about the behavior of the solution. As an example, in what follows, we find two exact solutions, and we explore their cosmological consequences.

In the first example we study the case $f(R,T)=f_{1}(R)+f_{2}(T)=R^{n}+\lambda
T,$ while in the second example we investigate the case $f(R,T)=f_{2}%
(R)f_{3}(T)=R^{m}T^{l}$.

\subsection{Case 1: $f(R,T)=f_{1}(R)+f_{2}(T)$}

In this case, the FE are:%
\begin{equation}
G_{\mu\nu}=\frac{1}{f_{1R}}\left(  \left(  8\pi+f_{2T}\right)  T_{\mu\nu
}+f_{2T}pg_{\mu\nu}+\frac{1}{2}f_{2}g_{\mu\nu}+\frac{1}{2}\left(  f_{1}%
-f_{1R}R\right)  g_{\mu\nu}-\left(  g_{\mu\nu}\square-\nabla_{\mu}\nabla_{\nu
}\right)  f_{1R}\right)  ,
\end{equation}
\begin{equation}
\left(  1+\frac{2f_{T}}{8\pi-f_{T}}\right)  \nabla^{\mu}T_{\mu\nu}=\frac
{f_{T}}{8\pi-f_{T}}\left[  \left(  -T_{\mu\nu}-pg_{\mu\nu}\right)  \nabla
^{\mu}\ln f_{T}-\nabla^{\mu}\left(  pg_{\mu\nu}\right)  -\frac{1}{2}g_{\mu\nu
}\nabla^{\mu}T\right]  ,
\end{equation}
and we may reformulate them as an effective Einstein field equations of the form%
\begin{equation}
G_{\mu\nu}=G_{\mathrm{eff}}T_{\mu\nu}+T_{\mu\nu}^{\mathrm{eff}},\qquad
G_{\mathrm{eff}}=\frac{1}{f_{1R}}\left(  1+f_{2T}\right)  ,
\end{equation}
where $G_{\mathrm{eff}}$ is the effective gravitational coupling. In fact, it
is a time varying coupling, while $T_{\mu\nu}^{\mathrm{eff}}$ is defined as%
\begin{equation}
T_{\mu\nu}^{\mathrm{eff}}=\frac{1}{f_{1R}}\left(  \frac{1}{2}\left(
2f_{2T}p+f_{2}+\left(  f_{1}-f_{1R}R\right)  \right)  g_{\mu\nu}-\left(
g_{\mu\nu}\square-\nabla_{\mu}\nabla_{\nu}\right)  f_{1R}\right)  .
\end{equation}

Under the above assumptions about the matter field, and taking into account
the previous results (\ref{R1}), that is:%
\begin{equation}
f(R,T)=f_{1}(R)+f_{2}(T)=R^{n}+\lambda T,
\end{equation}
where $R=R_{0}t^{-2},$ $a=t^{a_{1}},$ $a_{1}\in\mathbb{R}^{+},$ and
$p=\omega\rho=p_{0}t^{-2n},$ $T=-\rho+3p=T_{0}t^{-2n},$ with $n,\lambda
\in\mathbb{R},$ then, the FE and the conservation equation for a flat FLRW
metric read:%
\begin{align}
3f_{1R}H^{2}  &  =\left(  1+f_{2T}\right)  \rho-f_{2T}p-\frac{1}{2}%
f_{2}-3f_{1R}^{\prime}H+\frac{1}{2}\left(  f_{1R}R-f_{1}\right)
,\label{c1FE1}\\
f_{1R}\left(  2H^{\prime}+3H^{2}\right)   &  =-\left(  1+2f_{2T}\right)
p-\frac{1}{2}f_{2}-f_{1R}^{\prime\prime}-2f_{1R}^{\prime}H+\frac{1}{2}\left(
f_{1R}R-f_{1}\right)  , \label{c1FE2}%
\end{align}%
\begin{equation}
\left[  1+\frac{1}{2}\left(  3-5\omega\right)  f_{2T}\right]  \rho^{\prime
}+\left[  3(1+f_{2T})\left(  \omega+1\right)  H+\left(  1-\omega\right)
f_{2T}^{\prime}\right]  \rho=0. \label{c1FE3}%
\end{equation}

Starting from the results (\ref{R1}), we have $f_{1R}=nR^{n-1}$ ;
$f_{1R}^{\prime}=-2n(n-1)R_{0}^{n-1}t^{-2n+1}$ and $f_{1R}^{\prime\prime
}=2n(n-1)(2n-1)R_{0}^{n-1}t^{-2n}$.

Additionally, from (\ref{R1}) we have the relation $T=T_{0}t^{-2n}$ , with
$T_{0}=\rho_{0}(3\omega-1).$ By inserting these expressions into Eqs.
(\ref{c1FE1}-\ref{c1FE3}) we obtain%
\begin{equation}
\rho_{0}\left[  1+\frac{\lambda}{2}(3-5\omega)\right]  =A,\qquad a_{1}%
=\frac{n\left(  2-\lambda\left(  5\omega-3\right)  \right)  }{3\left(
\omega+1\right)  \left(  \lambda+1\right)  }, \label{Eq18}%
\end{equation}
$\allowbreak$where%
\begin{equation}
A\equiv R_{0}^{n-1}\left(  3na_{1}^{2}+\frac{1}{2}(1-n)\left(  R_{0}%
+12na_{1}\right)  \right)  ,\qquad R_{0}=6a_{1}(2a_{1}-1),
\end{equation}
and therefore%
\begin{equation}
G_{\mathrm{eff}}=\frac{\left(  1+\lambda\right)  }{nR_{0}^{n-1}}t^{2\left(
n-1\right)  }.
\end{equation}

If $\lambda=0$ and $n=1,$ from Eq. (\ref{Eq18}) one recovers the Standard
$\Lambda$CDM Model relations%
\begin{equation}
a_{1}=\frac{2}{3\left(  \omega+1\right)  },\qquad\rho_{0}=3a_{1}^{2},\qquad
G_{\mathrm{eff}}=const.
\end{equation}

In order to study the solution described in Eq. (\ref{Eq18}), that is, to
determine the range of the parameters $\left(  \omega,n,\lambda\right)  ,$ we
have performed a numerical analysis for two cases.

In Case 1, we have fixed $n=1$ (the Standard Model) and in Case 2, we have
fixed $n=2.$ It is possible to extend this analysis to any value of $n,$ that
is, $n=3/2,$ $n=3,$ etc.

\textbf{Case 1a}, $n=1.$ To estimate the $\lambda$-effect on Standard Model,
we focus to the case $n=1$. The physical quantities behave as%
\begin{equation}
a=t^{a_{1}},\qquad a_{1}=\frac{\left(  2-\lambda\left(  5\omega-3\right)
\right)  }{3\left(  \omega+1\right)  \left(  \lambda+1\right)  },\qquad
\rho=\rho_{0}t^{-2},\qquad\rho_{0}=\frac{3a_{1}^{2}}{1+\frac{\lambda}%
{2}(3-5\omega)},
\end{equation}
while the effective gravitational constant becomes a constant $G_{\mathrm{eff}%
}=const,$ but under the restriction $\lambda>-1.$

\begin{figure}[h]
\begin{center}
\includegraphics[height=1.5in,width=1.5in]{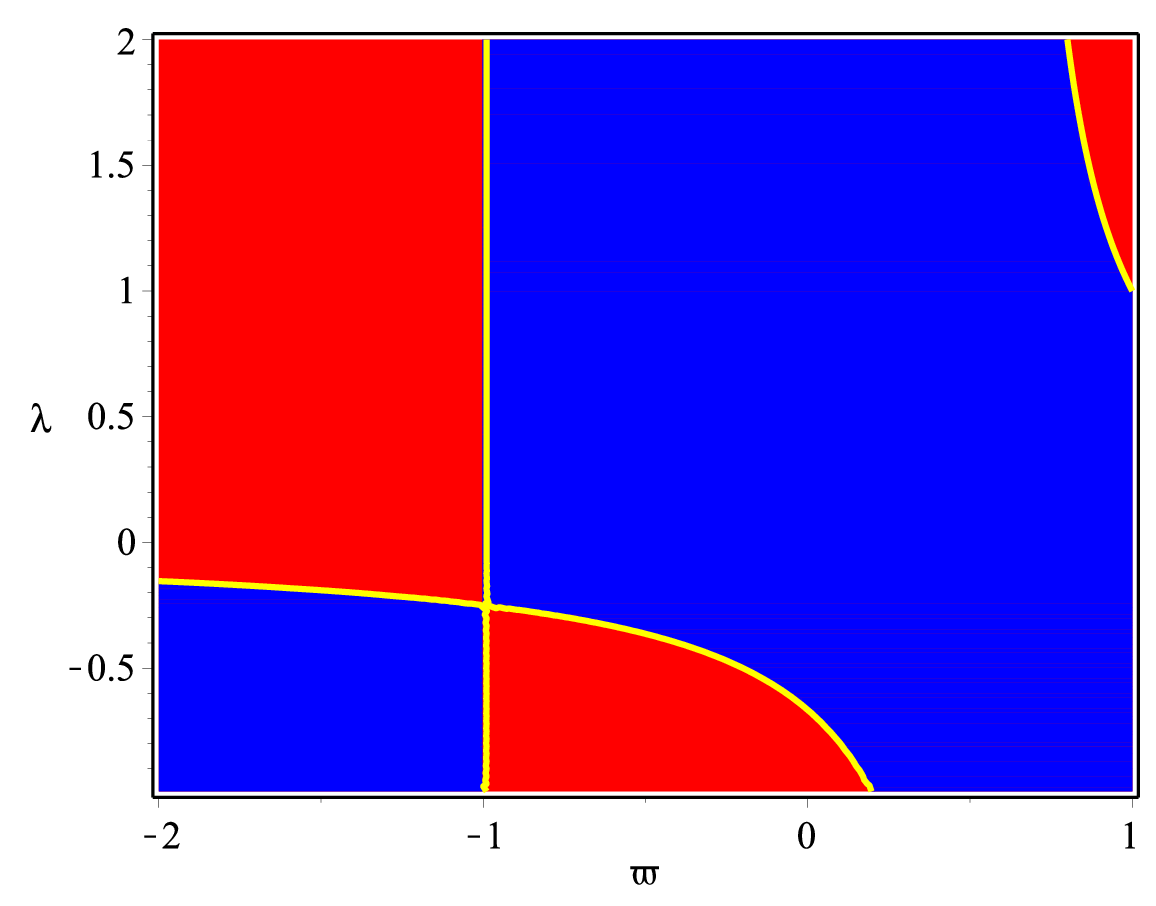}
\includegraphics[height=1.5in,width=1.5in]{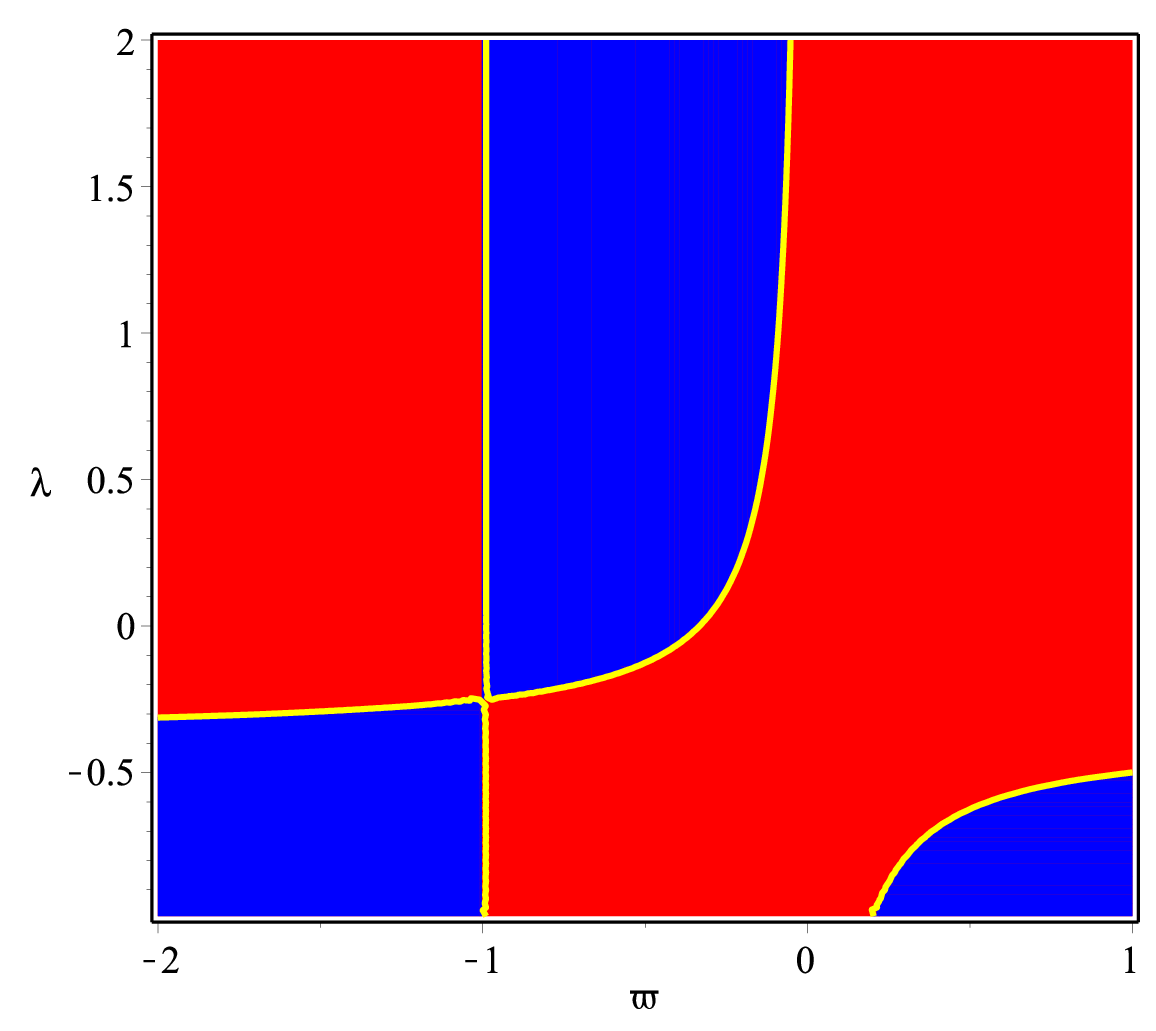}
\includegraphics[height=1.5in,width=1.5in]{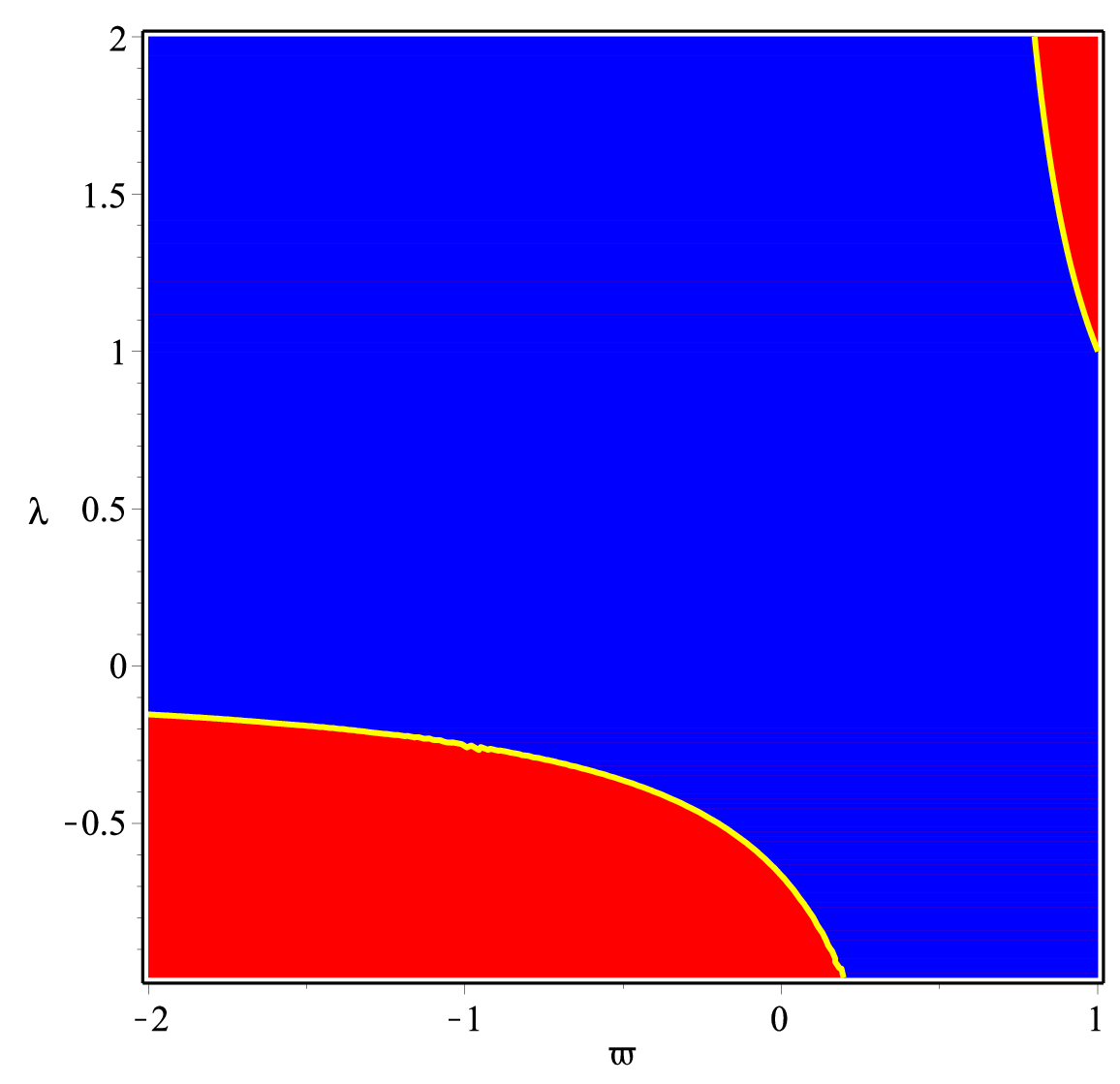}
\end{center}
\caption{Case 1a $n=1.$ Plots of $a_{1}$ (blue area $a_{1}>0,$ red color
$a_{1}\leq0$). Second figure, we show the region where $a_{1}>1$ (blue color)
red color means $a_{1}\leq1.$ Last figure, $\rho_{0}.$ Blue color $\rho_{0}%
>0$, while red color $\rho_{0}\leq0.$}%
\label{Case1}%
\end{figure}

In Figs.~\ref{Case1} we have performed a numerical analysis of the quantities
$a_{1}$ and $\rho_{0},$ by plotting these quantities in the plane $\left(
\omega,\lambda\right)  .$ We have fixed $\omega\in\left[  -2,1\right]  $
$\ $and $\lambda\in\left[  -2,2\right]  $, noting that the analysis may be
enlarged by considering $\lambda\in\left[  -m,m\right]  ,$ $m\in\mathbb{R},$
since, a priori, there is no any restriction on this numerical constant.

We have started the analysis by plotting the area within the region $\left(
\mathcal{R}=\left[  -2,1\right]  \times\left[  -2,2\right]  \right)  $\ where
$a_{1}>0$ (blue color) while red color means $a_{1}\leq0.$ In the mid
sub-panel, we have calculated the region where $a_{1}>1$ (blue color), that
is, we have calculated those values of $\omega$ and $\lambda$ for which the
solution represents an accelerated expansion (red color means $a_{1}\leq1$).
In the right sub-panel, we have calculated the values of $\left(
\omega,\lambda\right)  $ for which $\rho_{0}>0$ (blue color) while, as in the
above pictures, red color means $\rho_{0}\leq0.$

We may conclude, roughly speaking, that the solution has physical meaning if
$\omega\in(-1,1]$ and $\lambda\in\left[  -2,2\right]  .$ Note that $a_{1}$
tends to infinity when $\omega\rightarrow1,$, and $\lambda\rightarrow-1.$ If
$\omega\in(-1,0]$, then the solution accelerates. Only by taking into account
the condition $G_{\mathrm{eff}}=const>0,$ we may find the following constraint
on $\lambda:$ $\lambda>-1.$ Thus, our solution works well in the region
$\mathcal{R}=(-1,1]\times(-1,2].$ From the observations, the deceleration
parameter $q$ is measured as $q=-0.18_{-0.12}^{+0.12},$ thus this means that
$a_{1}\in\left[  1.064,1.429\right]  .$

\textbf{Case 1b}, $n=2.$ In order to estimate the \emph{curvature}-effect on
the Standard Model, we focus on the case $n=2$. The physical quantities
behaves as%
\begin{equation}
a=t^{a_{1}},\qquad a_{1}=\frac{2\left(  2-\lambda\left(  5\omega-3\right)
\right)  }{3\left(  \omega+1\right)  \left(  \lambda+1\right)  },\qquad
\rho=\rho_{0}t^{-4},\qquad\rho_{0}=\frac{R_{0}\left(  12a_{1}^{2}-\left(
R_{0}+24a_{1}\right)  \right)  }{2+\lambda(3-5\omega)},
\end{equation}
while the effective gravitational constant behaves as%
\begin{equation}
G_{\mathrm{eff}}=\frac{\left(  1+\lambda\right)  }{2R_{0}}t^{2}=\frac{3\left(
1+\lambda\right)  ^{3}\left(  \omega+1\right)  ^{2}}{8\left(  \lambda\left(
5\omega-3\right)  -2\right)  \left(  3\omega-5+\lambda\left(  23\omega
-9\right)  \right)  }t^{2},
\end{equation}
so it is a positive function, if $\lambda>-1$. Note that for these numerical
values it is an increasing function.

\begin{figure}[h]
\begin{center}
\includegraphics[height=1.5in,width=1.5in]{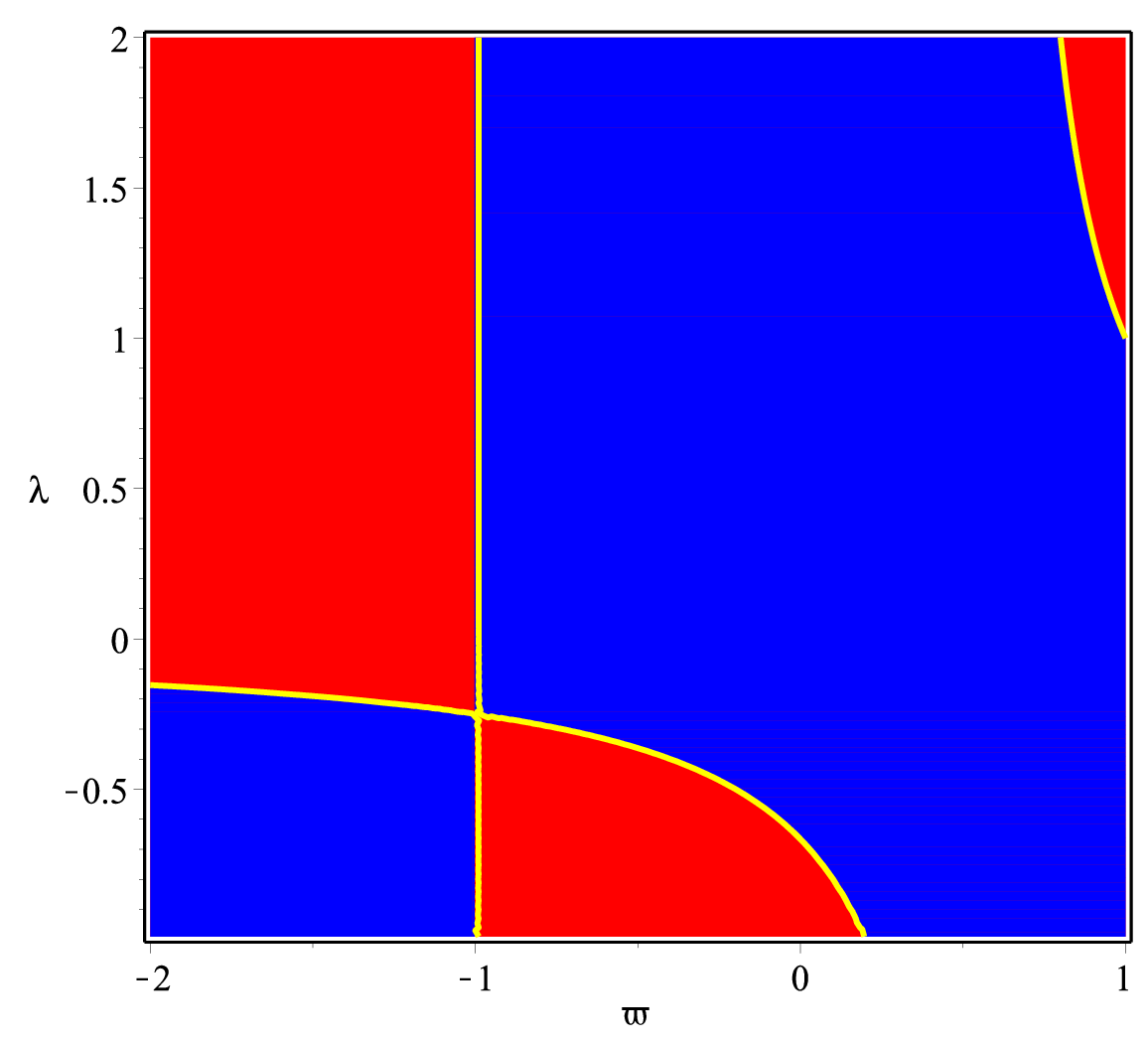}
\includegraphics[height=1.5in,width=1.5in]{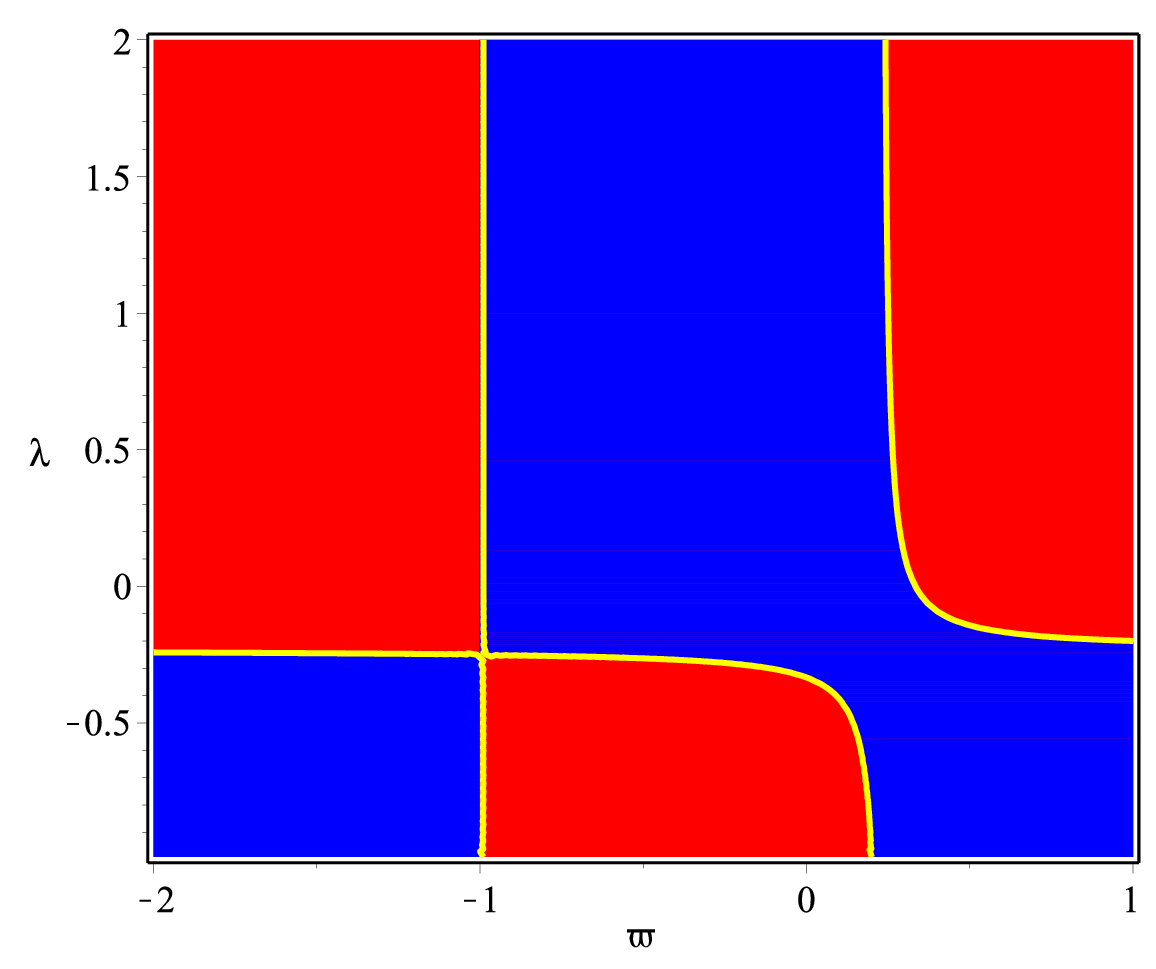}
\includegraphics[height=1.5in,width=1.5in]{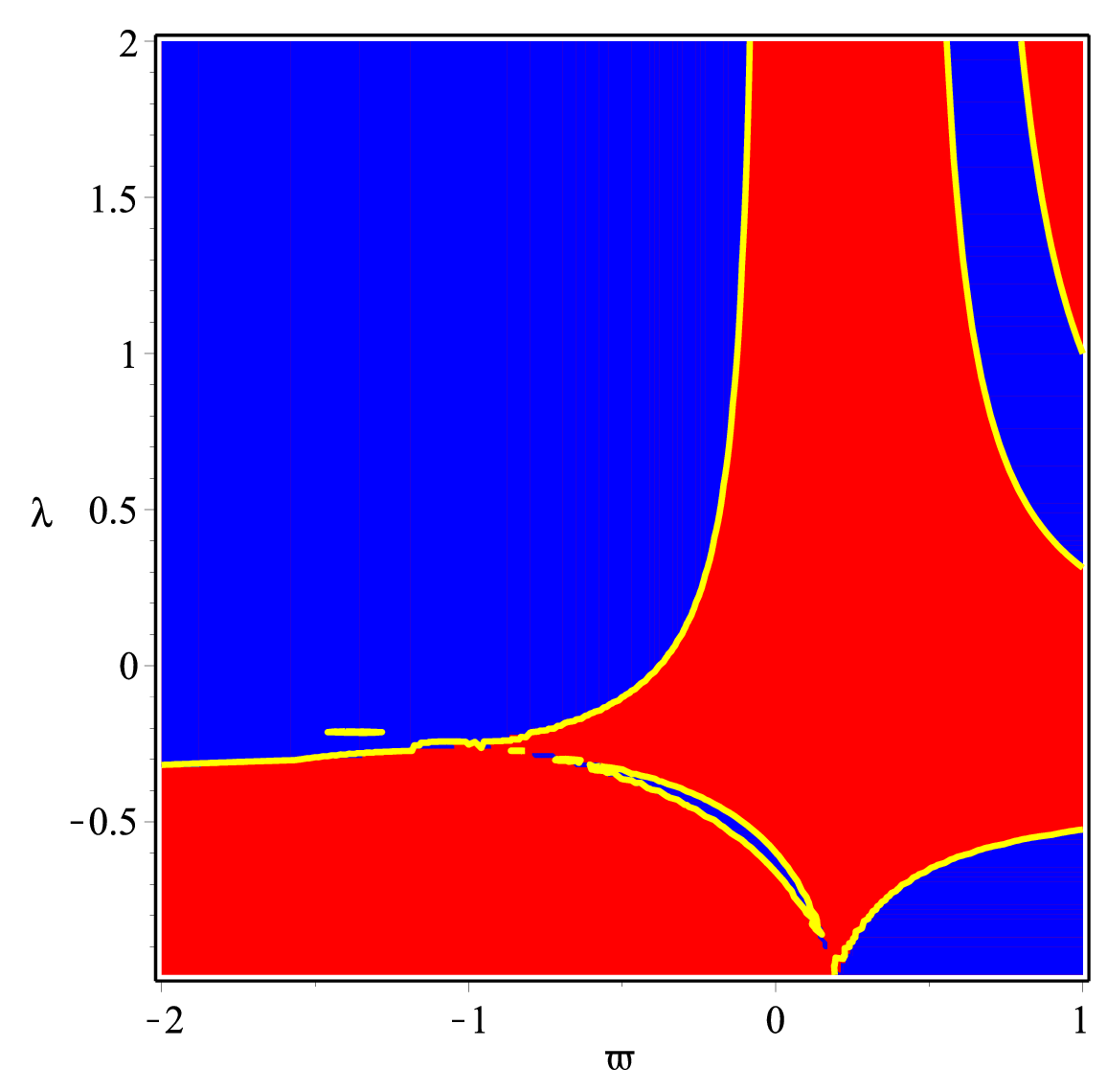}
\end{center}
\caption{Case 1b $n=2.$ Plots of $a_{1}$ (blue area $a_{1}>0,$ red color
$a_{1}\leq0$). Second figure, we show the region where $a_{1}>1$ (blue color)
red color means $a_{1}\leq1.$ Last figure, $\rho_{0}.$ Blue color $\rho_{0}%
>0$, while red color $\rho_{0}\leq0.$}%
\label{Case2}%
\end{figure}

By carrying out a similar analysis as in the previous case, we have fixed a
region $\mathcal{R}$ in the plane $\left(  \omega,\lambda\right)  ,$ with
$\mathcal{R}=\left[  -2,1\right]  \times\left[  -2,2\right]  .$ Thus, in
Figs.~(\ref{Case2}), we have calculated the area within the region
$\mathcal{R}$ where the exponent of the scale factor, $a_{1}$ is positive
(left sub panel) and where it is $a_{1}>1$ (mid sub panel). In the right
sub-panel of Fig. (\ref{Case2}) we have plotted the area where $\rho_{0}>0$
(blue color) while, as in the above pictures, red color means $\rho_{0}\leq0.$

As we can see from this last picture, basically $\omega\in(-1,0]$ where
furthermore $a_{1}>1,$ and $\lambda\in(-1,2],$ in such a way that in this
region, $\mathcal{R}=(-1,0]\times(-1,2],$ $G_{\mathrm{eff}}=const>0.$

\subsection{Case 2: $f(R,T)=f_{1}(R)f_{2}(T)$}

For this model the FE reduce to%
\begin{equation}
G_{\mu\nu}=G_{\mathrm{eff}}T_{\mu\nu}+T_{\mu\nu}^{\mathrm{eff}},
\end{equation}%
\begin{equation}
G_{\mu\nu}=\frac{1}{K}\left(  \left(  8\pi+f_{T}\right)  T_{\mu\nu}%
+f_{T}pg_{\mu\nu}+\frac{1}{2}\left(  f-KR\right)  g_{\mu\nu}-\left(  g_{\mu
\nu}\square-\nabla_{\mu}\nabla_{\nu}\right)  K\right)  ,
\end{equation}
with $f_{R}=f_{1R}f_{2}=K,$ and $f_{T}=f_{2T}f_{1}$, where%
\begin{equation}
G_{\mathrm{eff}}=\frac{8\pi+f_{2T}f_{1}}{f_{1R}f_{2}}.
\end{equation}

The conservation equation becomes%
\begin{equation}
\left(  1+2C\right)  \nabla^{\mu}T_{\mu\nu}=-C\left[  \left(  T_{\mu\nu
}+pg_{\mu\nu}\right)  \nabla^{\mu}\ln\left(  f_{1}f_{2T}\right)  +\frac{5}%
{2}g_{\mu\nu}\nabla^{\mu}p-\frac{1}{2}g_{\mu\nu}\nabla^{\mu}\rho\right]  ,
\end{equation}
where $C=\frac{f_{1}f_{2T}}{1-f_{1}f_{2T}}.$

For a flat FLRW metric, the FE and the conservation equation read%
\begin{align}
3KH^{2}  &  =\rho\left(  1+f_{1}f_{2T}\right)  -pf_{1}f_{2T}-3K^{\prime
}H+\frac{1}{2}\left(  KR-f_{1}f_{2}\right)  ,\label{case21}\\
K\left(  2H^{\prime}+3H^{2}\right)   &  =-p\left(  1+2f_{1}f_{2T}\right)
-K^{\prime\prime}-2K^{\prime}H+\frac{1}{2}\left(  KR-f_{1}f_{2}\right)  ,
\label{case22}%
\end{align}%
\begin{equation}
\left(  1+\frac{5}{2}\left(  1-\omega\right)  C\right)  \rho^{\prime}=\left[
C\left(  \omega-1\right)  \left(  \frac{f_{2T}^{\prime}}{f_{2T}}+\frac
{f_{1}^{\prime}}{f_{1}}\right)  -3\left(  \omega+1\right)  H\left(
1+2C\right)  \right]  \rho. \label{case23}%
\end{equation}

By taking into account the previous results (\ref{R2a}),
\begin{equation}
f(R,T)=f_{1}(R)f_{2}(T)=\lambda R^{m}T^{l},\qquad2m=x\left(  1-l\right)  ,
\end{equation}
where%
\begin{equation}
f_{1}(R)=R^{m}=R_{0}t^{-2m},\qquad f_{2}(T)=\lambda T^{l}=\lambda T_{0}%
t^{-xl},\qquad\lambda,m,l,x\in\mathbb{R},
\end{equation}
and recalling that $\rho\thickapprox p\thickapprox T\thickapprox t^{-x},$
$\left(  T_{0}=\rho_{0}(3\omega-1)\right)  $, we obtain
\begin{equation}
a=t^{a_{1}},\qquad a_{1}\in\mathbb{R}^{+},\qquad R=6\left(  \frac
{a^{\prime\prime}}{a}+H^{2}\right)  =6a_{1}(2a_{1}-1)t^{-2}=R_{0}t^{-2}.
\end{equation}

Note that if we set $x=2,$ then $m+l=1,$ since $f_{1}f_{2T}=const,$ the
Standard Model is recovered if we set: $m=1$ and $l=0$ $\left(  \lambda
=1\right)  .$

By solving Eqs. (\ref{case21}-\ref{case23}) we find
\begin{equation}
\rho_{0}=\left[  6a_{1}(2a_{1}-1)\right]  ^{\frac{x}{2}}\left(  \frac
{m\lambda(3\omega-1)^{1-\frac{2m}{x}}}{6a_{1}(2a_{1}-1)(1+\omega)}\left[
2a_{1}+(2-x)(a_{1}+x-1)\right]  -\frac{\lambda(1-\frac{2m}{x})}{(3\omega
-1)^{\frac{2m}{x}}}\right)  ^{\frac{x}{2m}}, \label{r}%
\end{equation}
and
\begin{equation}
m=\frac{\frac{\omega}{3\omega-1}+\frac{1}{2}}{B+\frac{2\omega}{x(3\omega
-1)}+\frac{1}{2}}, \label{m}%
\end{equation}
with
\begin{equation}
B\equiv\frac{a_{1}\left(  2-3a_{1}(1+\omega)\right)  +(2-x)\left(
x-1-a_{1}(2+3\omega)\right)  }{6a_{1}(2a_{1}-1)(1+\omega)}.
\end{equation}

We may obtain an expression for $a_{1}$ from Eq. (\ref{m}). It must satisfy
the following quadratic equation
\begin{equation}
c_{1}a_{1}^{2}+c_{2}a_{1}+c_{3}=0, \label{a1}%
\end{equation}
where%
\begin{align}
c_{1}  &  =3\left(  1-\frac{2}{m}\right)  -\frac{12\omega}{x\left(
3\omega-1\right)  }\left(  \frac{x}{m}-2\right)  ,\\
c_{2}  &  =\frac{2+\left(  x-2\right)  \left(  2+3\omega\right)  }{1+\omega
}+\frac{6\omega}{x\left(  3\omega-1\right)  }\left(  \frac{x}{m}-2\right)
+3\left(  \frac{1}{m}-1\right)  ,\\
c_{3}  &  =\frac{\left(  x-2\right)  \left(  x-1\right)  }{1+\omega},
\end{align}
so we obtain two solutions for $a_{1}=a_{1\pm}.$ Note that, in the Standard
Model case, the above Eqs. are reduced to%
\begin{equation}
\rho_{0}=\frac{2a_{1}}{1+\omega},\qquad m=1,\qquad a_{1}=\frac{2}{3\left(
1+\omega\right)  }.
\end{equation}

Therefore, we have solved completely our proposed model. The constants
$\rho_{0}$ and $a_{1}$ depend on $\left(  m,x,\omega\right)  $ with the
relationship $2m=x\left(  1-l\right)  .$ Since $m$ and $x$ are positive, then
$l<1.$ Thus, our model is described for the function
\begin{equation}
f(R,T)=f_{1}(R)f_{2}(T)=\lambda R^{m}T^{l},\qquad l<1,\qquad2m=x\left(
1-l\right)  ,\qquad\lambda\in\mathbb{R}.
\end{equation}

We may study the behavior of this solution in several ways. For example,
by fixing $m$ so the constants $\rho_{0}$ and $a_{1}$ depend on $\left(
x,\omega\right)  $ or by fixing $l.$

We have explored the case $l=-1$ which implies that $m=x$ $\left(
f=R^{x}T^{-1}\right)  ,$ in such a way that Eqs. (\ref{r} and \ref{a1}) are
reduced to%
\begin{align*}
\rho_{0}  &  =\left[  6a_{1}(2a_{1}-1)\right]  ^{\frac{x}{2}}\left(
\frac{x\lambda\left(  2a_{1}+(2-x)(a_{1}+x-1)\right)  }{6a_{1}(2a_{1}%
-1)(1+\omega)(3\omega-1)}+\frac{\lambda}{(3\omega-1)^{2}}\right)  ^{\frac
{1}{2}}\\
c_{1}a_{1}^{2}+c_{2}a_{1}+c_{3}  &  =0\Longrightarrow a_{1}=a_{1\pm}%
\end{align*}
$\allowbreak$where%
\begin{align}
c_{1}  &  =3\left(  1-\frac{2}{x}\right)  +\frac{12\omega}{x\left(
3\omega-1\right)  },\qquad c_{3}=\frac{\left(  x-2\right)  \left(  x-1\right)
}{1+\omega},\\
c_{2}  &  =\frac{2+\left(  x-2\right)  \left(  2+3\omega\right)  }{1+\omega
}-\frac{6\omega}{x\left(  3\omega-1\right)  }+3\left(  \frac{1}{x}-1\right)  ,
\end{align}
and%
\begin{equation}
G_{\mathrm{eff}}=\frac{1+f_{T}}{f_{R}}=\left(  \frac{R_{0}^{1-x}}{x\lambda
T_{0}}\left(  T_{0}^{2}-\lambda R_{0}^{x}\right)  \right)  t^{x-2}%
=G_{\mathrm{eff}_{0}}t^{x-2}.
\end{equation}

To analyse this solution we have carried out a numerical study, as in the
above cases. By fixing a region $\mathcal{R}$ in the plane $\left(
\omega,x\right)  ,$ with $\mathcal{R}=(-1,1]\times\left[  1,3\right]  ,$ then
we may see where the solution has a physical meaning. Thus, in Figs.~(\ref{Case3})
we have calculated the area within the region $\mathcal{R}$ where the exponent
of the scale factor, $a_{1}$ is positive (left sub panel) and where it is
$a_{1}>1$ (mid sub panel).

In the right sub-panel of Fig.~(\ref{Case3}) we have plotted the area where
$\rho_{0}>0$ (blue color) while, as in the above pictures, red color means
$\rho_{0}\leq0.$ As we can see from this last picture (right sub panel),
basically, the parameter $\rho_{0}$ is not positive for all $\omega\in(-1,0],$
thus, this solution looks restrictive. With regard to the behavior of
$G_{\mathrm{eff}},$ we may say that it behaves as an increasing time function
if $x>2$, constant if $x=2$ and a decreasing function if $x<2.$

\begin{figure}[h]
\begin{center}
\includegraphics[height=1.5in,width=1.5in]{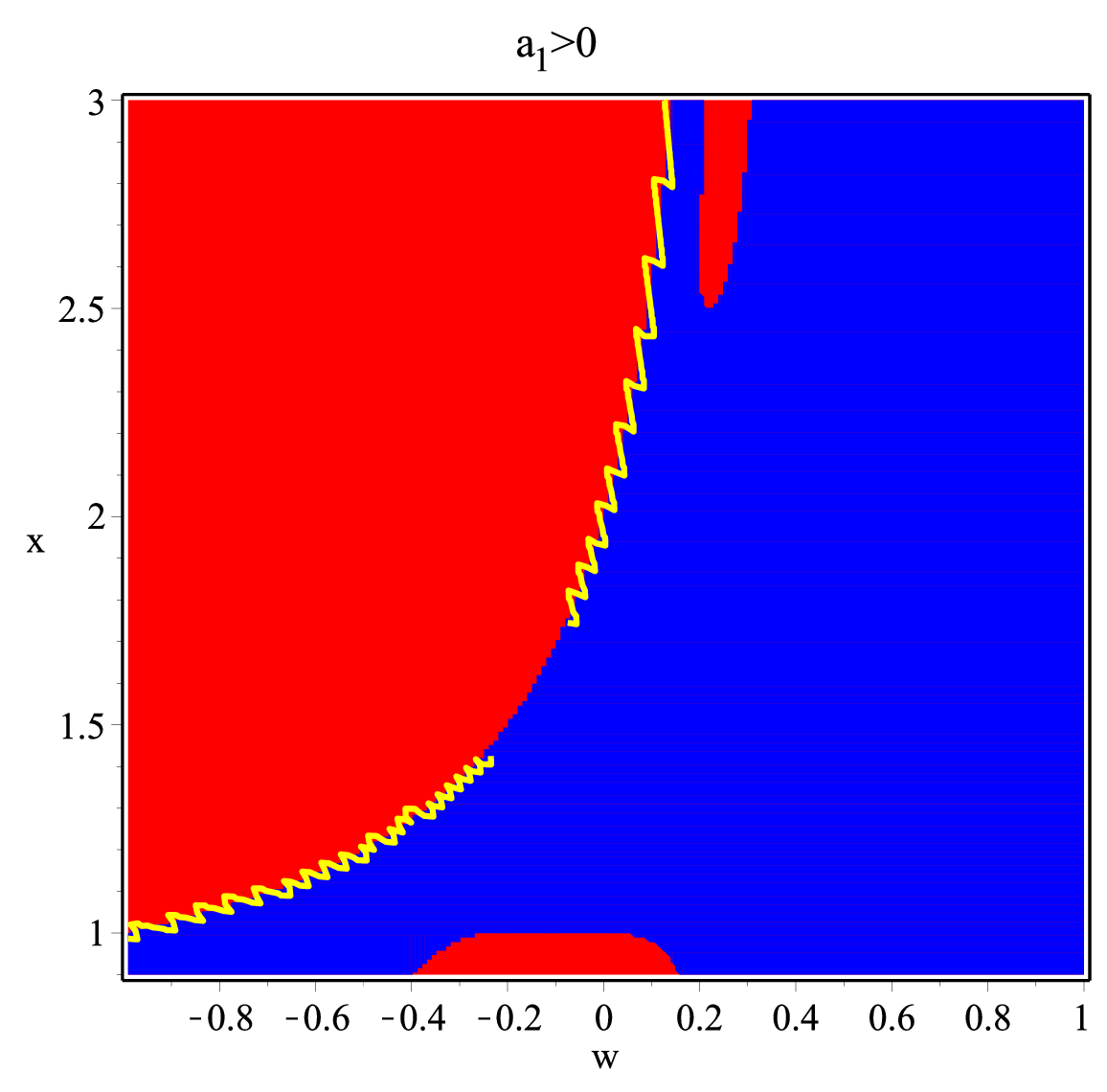}
\includegraphics[height=1.5in,width=1.5in]{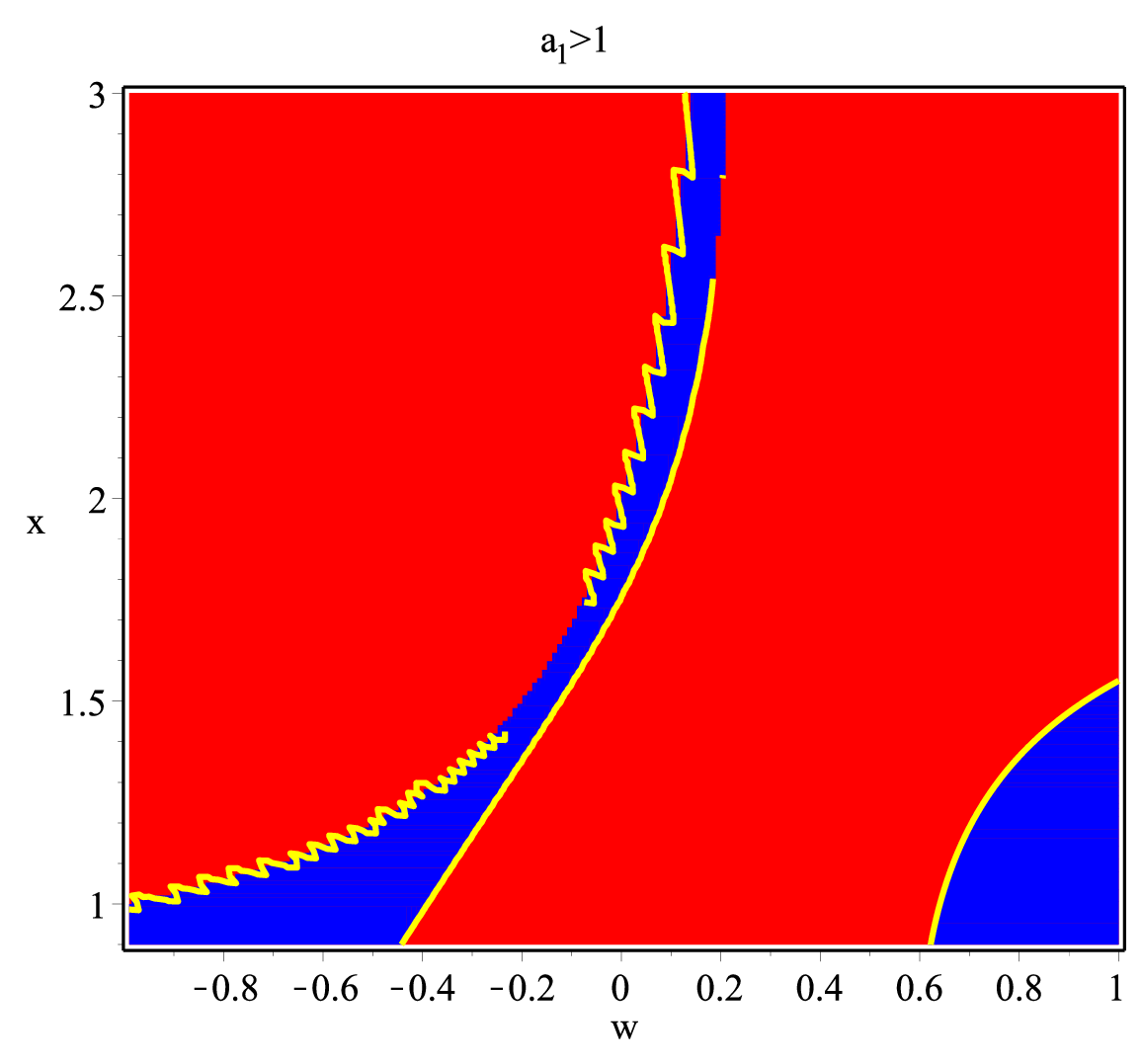}
\includegraphics[height=1.5in,width=1.5in]{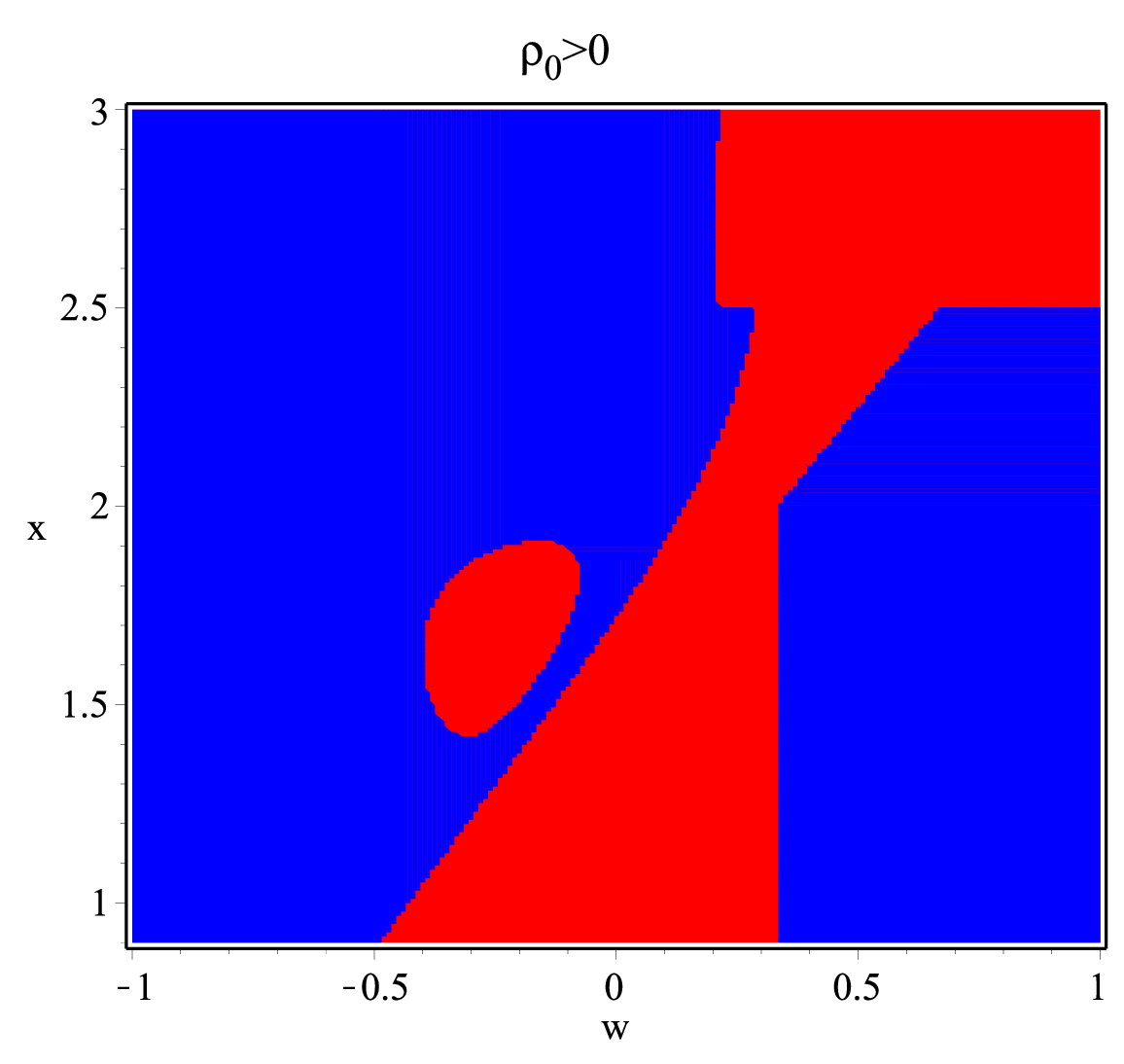}
\end{center}
\caption{Case 2 $f(R,T)=f_{1}(R)f_{2}(T)$. In the left sub panel, we plotted
$a_{1}$ (blue area $a_{1}>0,$ red color $a_{1}\leq0$). In the mid sub panel,
we show the region where $a_{1}>1$ (blue color), red color means $a_{1}\leq1.$
Last figure, $\rho_{0}.$ Blue color $\rho_{0}>0$, while red color $\rho
_{0}\leq0.$}%
\label{Case3}%
\end{figure}

As we have mentioned above, it is possible to analyse the solution
(\ref{r}, \ref{m} and \ref{a1}) in several ways. This could be achieved by
following a different approach. For example, fixing $x=2$, leads to
$m=(1-l)$, and this could result in a less restrictive solution.

\section{Conclusions}

In the present paper, we have explored whether the gravitational model $f(R,T)$
admits self-similar solutions, or not. We have considered only two cases for
the function $f(R,T).$ The first studied case corresponds to $f(R,T)=f_{1}%
(R)+f_{2}(T)$ while the second one is $f(R,T)=f_{1}(R)+f_{2}(R)f_{3}(T).$ In
order to determine if these models admit self similar solutions, we have used
the geometrical method of the matter collineations, by considering that the
matter collineation vector field is the homothetic one. This approach allows
us to determine the exact form of the unknown functions $f_{1}(R),f_{2}(T),$
as well as the rest of the physical and geometrical quantities.

For simplicity, we have considered a perfect fluid, and we have determined the
behavior of the energy density $\rho.$ The obtained Theorems are quite
general, and valid not only for the FLRW metric, but also valid for all the
self-similar Bianchi models and the Kantowski-Sachs one.

In the first of the studied cases, with $f(R,T)=f_{1}(R)+f_{2}(T),$ we have
been able to determine that the FE admit self-similar solutions if $f_{1}=R^{n},$ and
$f_{2}=\lambda T,$ with $n,\lambda\in\mathbb{R},$ therefore $f(R,T)=R^{n}%
+\lambda T$ (see \ref{Teorema1}). This new result generalizes the model
proposed by Harko et al. in \cite{H1}. In the Corollary (\ref{Colo1}), we show
that the model $f(R,T)=R+\lambda T$, admits self-similar solutions as well.

In the second of the studied cases: $f(R,T)=f_{1}(R)+f_{2}(R)f_{3}(T),$ we
have stated in Theorem (\ref{Teorema2}) that the FE admit self-similar
solutions, if $f(R,T)=R^{n}+R^{n(1-l)}T^{l}$. As we have pointed out, if we
fix $l=1,$ then the model reduces to the first studied case, that is, the case
$f(R,T)=f_{1}(R)+f_{2}(T)$. From the Theorem (\ref{Teorema2}) we have the
Corollaries (\ref{Coro2} and \ref{Coro3}), which allow us to obtain two
particular cases: $f(R,T)=R+R^{m}T^{l},$ such that $m+l=1$ and $f(R,T)=R^{m}%
T^{l},$ under the constraint $x\left(  1-l\right)  =2m.$ All these results are
also new, leading to novel classes of cosmological models.

Furthermore, we have shown how to calculate the exact solution in two
particular cases, $f(R,T)=R^{n}+\lambda T$, and $f(R,T)=R^{m}T^{l}$, respectively.

The determinations by the Planck satellite of the Cosmic Microwave Background Radiation temperature fluctuations \cite{1g,1h}, together with the observations of the light curves of the distant supernovae \cite{Riess} have provided a powerful confirmation of the amazing fact that the Universe is in a phase experiencing a de Sitter type accelerating expansion. Moreover, other amazing observational result and that the matter composition of the Universe consists of only 5\% baryonic matter, while 95\% of matter-energy is represented by two mysterious forms of energy/matter, called dark energy and dark matter, respectively. To find an explanation for the cosmological observational results, the $\Lambda$CDM paradigm was introduced, which is essentially based by the introduction in the Einstein gravitational field equations of the cosmological constant $\Lambda$. The cosmological constant was introduced in general relativity by Einstein in 1917 \cite{Ein}, in order to construct a static cosmological model. The $\Lambda$CDM model gives and excellent fit to the observational data, especially at low redshifts. However, its  theoretical basis is questionable, since there are no convincing answers to the many problems raised by physical nature and interpretation of $\Lambda$.

Therefore, to obtain a physically and mathematically acceptable description of the Universe, one must go beyond standard general relativity. In the present paper we have adopted the dark gravity approach, in which it is assumed that the nature of the gravitational interaction drastically changes on astrophysical (galactic) and cosmological scales. Therefore, the standard general relativistic Einstein equations, which offer a very precise description of the physics at the level of the Solar System, must be change by a new theory of gravity.  In the present paper we have assumed that such a theory is represented by the $f(R,T)$ theory, which implies a geometry-matter coupling \cite{H1}. Such an approach leads to gravitational models more complicated than standard general relativity.  Cosmological models built in the framework of $f(R,T)$ gravity represent an interesting possibility of explaining dark energy, the accelerating expansion of the Universe, and perhaps even dark matter. However, this type of theory also raises a number of extremely difficult  mathematical problems, and the present paper represents a systematic approach for investigating some of the properties of the theory, as well as proposing a rigorous mathematical algorithm for obtaining cosmological solutions.

From a cosmological, as well as a physical point of view, the relevance of the
obtained results is twofold. Firstly, it gives the possibility of obtaining
the functional forms of $f(R,T)$ that admit cosmological power law solutions
of the form $a(t)=t^{a_{1}}$, which are relevant for several phases of the
evolution of the Universe. The deceleration parameter corresponding to this
solution is given by $q=1/a_{1}-1$, and it can describe accelerating
cosmological models for $a_{1}>1$, decelerating evolution for $a_{1}<1$, and
marginally accelerating solutions for $a_{1}=1$. Hence, one could obtain
severe constraints for the allowed functional form of the $f(R,T)$ function by
requiring the presence in the theory of this power law form of the scale
factor. Generally, for this form of the scale factor, the geometrical (Ricci
scalar), and physical (energy density) parameters have a simple time
dependence, being inversely proportional to the square of the time. This
indicates that the considered self-similar solutions do present a singular
behavior at $t=0$. The effective gravitational coupling has a power law
dependence on time, and, for the $f(R,T)=f_{1}(R)+f_{2}(T)$ it increases as
$G_{eff}\propto t^{2}$.

The second implication of the self-similar approach is related to the
possibility of determining the integration constants that appear in the
theoretical models. Usually, these constants are determined from the initial
conditions, which are not, or poorly known. The present approach allows an
independent determination of the free parameters of the different cosmological
models, leading to the possibility of the direct confrontation of the
theoretical results with observations.

The self-similar approach for the investigation of the gravitational theories
represents a powerful approach that can give important information on the
mathematical and physical structures of the particular models. In the present
study we have performed such an analysis for the $f(R,T)$ modified gravity
theory, with the results of our analysis indicating that simple power law
cosmological models can be obtained in the framework of this modified gravity
theory. The cosmological implications of the present results will be fully
investigated in a future work.

\section*{Acknowledgement}

We would like to thank the anonymous reviewer for comments and suggestions
that helped us to significantly improve our work. J.A.B. is supported by
COMISI\'{O}N NACIONAL DE CIENCIAS Y TECNOLOG\'{I}A through FONDECYT Grant No 11170083.

\appendix

\section{Alternative approach}

For the model $f(R,T)=f_{1}(R)+f_{2}(T)$ FE read (for a flat FLRW metric, and
$T_{\mu\nu}$ a perfect fluid)
\begin{align}
3f_{1R}H^{2}  &  =\left(  1+f_{2T}\right)  \rho-f_{2T}p-\frac{1}{2}%
f_{2}-3f_{1R}^{\prime}H+\frac{1}{2}\left(  f_{1R}R-f_{1}\right)  ,\\
f_{1R}\left(  2H^{\prime}+3H^{2}\right)   &  =-\left(  1+2f_{2T}\right)
p-\frac{1}{2}f_{2}-f_{1R}^{\prime\prime}-2f_{1R}^{\prime}H+\frac{1}{2}\left(
f_{1R}R-f_{1}\right)  ,
\end{align}%
\begin{equation}
\left[  1+\frac{1}{2}\left(  3-5\omega\right)  f_{2T}\right]  \rho^{\prime
}+\left[  3(1+f_{2T})\left(  \omega+1\right)  H+\left(  1-\omega\right)
f_{2T}^{\prime}\right]  \rho=0.
\end{equation}

If one fixes the geometrical background as
\begin{equation}
a=e^{H_{0}t},
\end{equation}
this implies that
\begin{equation}
H=H_{0}=const.
\end{equation}
and therefore
\begin{equation}
R=R_{0}=const.
\end{equation}

Hence, in a naive approach (dimensional analysis point of view),
$f_{1}(R)=const.$ and the first of the FE give
\begin{align}
3f_{1R}H^{2}  &  =\left(  1+f_{2T}\right)  \rho-f_{2T}p-\frac{1}{2}%
f_{2}-3f_{1R}^{\prime}H+\frac{1}{2}\left(  f_{1R}R-f_{1}\right) \\
K_{1}  &  =\left(  1+f_{2T}\right)  \rho-f_{2T}p-\frac{1}{2}f_{2}-K_{2}+K_{3},
\end{align}
where $K_{i}$ are constants. To keep the dimensional homogeneity of the
equation, each of the remaining terms must be constant as well, that is,%
\begin{align}
\rho &  =\rho_{0},\Longrightarrow p=p_{0}\\
f_{2T}\rho &  =f_{2T}p=C_{1},\\
\frac{1}{2}f_{2}  &  =C_{2}.
\end{align}

However, $\rho=\rho_{0}$ is not a realistic physical result today. Note that
if one fix $f_{1}=f_{2}=1,$ then the FE are reduced to the standard FLRW FE.

To formalize this approach, we may to develop the following strategy. We
consider the following vector field (VF)%
\begin{equation}
X=K\partial_{t}-x\partial_{x}-y\partial_{y}-z\partial_{z}\in\mathfrak{X}%
(M),\qquad K\in\mathbb{R},
\end{equation}
then (for a FLRW metric) the condition%
\begin{equation}
L_{X}g_{\mu\nu}=0,
\end{equation}
yields%
\begin{equation}
Ka^{\prime}-a=0\qquad\Longleftrightarrow\qquad a=a_{0}e^{\frac{1}{K}t},
\end{equation}
that is, $a=\exp(H_{0}t).$

In the same way, if we calculate $L_{X}T_{\mu\nu}=0,$ it yields%
\begin{align}
\rho^{\prime}K  &  =0,\\
2Kpa^{\prime}+Kfp^{\prime}-2pf  &  =0,
\end{align}
that is%
\begin{align}
\rho &  =const.\\
2K\frac{a^{\prime}}{a}+K\frac{p^{\prime}}{p}  &  =2.
\end{align}
But, if we take into account that $a=a_{0}e^{\frac{1}{K}t}$, then%
\begin{equation}
2+K\frac{p^{\prime}}{p}=2,
\end{equation}
and therefore%
\begin{equation}
p^{\prime}K=0.
\end{equation}
Thus, $\rho=const=p$ and $f_{2}=const.$ as $f_{1}.$

\end{document}